\documentclass[a4paper,	11pt]{article}

\usepackage[english]{babel}
\usepackage[T1]{fontenc} 
\usepackage[babel=true]{csquotes} 
\usepackage{aeguill} 

\usepackage{latexsym}
\usepackage{comment}

\def \trans{^{\scriptscriptstyle{\intercal}}}

\usepackage[normalem]{ulem} 
\usepackage{hyperref}
\usepackage{tabularx}
\usepackage{amssymb}
\usepackage{amsmath,mathrsfs}
\usepackage{amsthm} 
\usepackage{graphicx}
\usepackage{empheq}
\usepackage{color}
\usepackage{empheq}
\usepackage{float}
\usepackage{caption}
\usepackage{subcaption}
\usepackage{xcolor}
\usepackage{adjustbox}
\usepackage{listings}
\usepackage{algorithm}
\usepackage{algpseudocode}
\usepackage[algo2e,ruled,vlined]{algorithm2e} 
\usepackage{algorithmicx}
\usepackage{bbm}

\definecolor{dkgreen}{rgb}{0,0.6,0}
\definecolor{gray}{rgb}{0.5,0.5,0.5}
\definecolor{mauve}{rgb}{0.58,0,0.82}

 \def \d{\mathrm{d}} 
 
\newcommand{\R}{{\bf R}}

\usepackage{amsmath, amssymb}
\usepackage{graphicx}
\usepackage{color}
\usepackage{textcomp}
\usepackage{graphics}
\usepackage{float}
\usepackage{empheq}
\usepackage{caption}
\usepackage{subcaption}

\usepackage{geometry}
\usepackage{stmaryrd}

\AtBeginDocument{\providecommand\propref[1]{\ref{prop:#1}}}
\AtBeginDocument{}

\newtheorem{thm}{Theorem}[section]

\newtheorem{assum}[thm]{Assumptions}

\newtheorem{prop}[thm]{Proposition}

\newtheorem{rem}[thm]{Remark}

\mathtoolsset{showonlyrefs}

\newcommand{\bm}{\bibitem}
\newcommand{\be}{\begin{equation}}
\newcommand{\ee}{\end{equation}}

\newcommand{\bea}{\begin{eqnarray}}
\newcommand{\bes}{\begin{subEquations}}
\newcommand{\ees}{\end{subEquations}}
\newcommand{\bgt}{\begin{gather}}
\newcommand{\egt}{\begin{gather}}
\newcommand{\eea}{\end{eqnarray}}
\newcommand{\beaa}{\begin{eqnarray*}}
\newcommand{\eeaa}{\end{eqnarray*}}

\newcommand{\EE}{{\mathbb E}}

\newcommand{\no}{\noindent}

\def \R{\mathbb{R}}

\def \E{\mathbb{E}}

\def \T{\mathbb{T}}

\def \Ac{{\cal A}}

\def \Pc{{\cal P}}

\def \Oc{{\cal O}}

\def \Tc{{\cal T}}

\def \Vc{{\cal V}}

\def \Xc{{\cal X}}

\def\red#1{{\color{red}#1}}

\def \mrJ{\mathrm{J}} 
\def \mra{\mathrm{a}}

\def \borho{\boldsymbol{\rho}}
\def \boa{\boldsymbol{a}}
\def \bolx{\boldsymbol{x}}

\numberwithin{equation}{section} 

\begin{document}

\title{Policy gradient learning methods for stochastic control with exit time and  applications to  share repurchase pricing   
}
\author{Mohamed HAMDOUCHE\footnote{LPSM, Université Paris Cité, hamdouche at lpsm.paris} \quad  
Pierre HENRY-LABORDERE\footnote{Qube Research and Technologies, pierre.henrylabordere at qube-rt.com} \quad  Huy\^en PHAM\footnote{LPSM, Université Paris Cité, pham at lpsm.paris}}

\date{}

\maketitle

\begin{abstract}  
We develop policy gradients methods for stochastic control with exit time in a model-free setting. 
We propose two types of algorithms for learning either directly the optimal policy or by learning alternately the value function (critic) and the optimal control (actor). 
The use of randomized policies is crucial for overcoming notably the issue related to the exit time in the gradient computation. We demonstrate the effectiveness of our approach by implementing our numerical schemes in  the application to the problem of share repurchase pricing. Our results 
show that the proposed policy gradient methods  outperform PDE or other neural networks techniques in a model-based setting. Furthermore, our algorithms are flexible enough to incorporate 
realistic market conditions like e.g.  price impact or transaction costs. 
\end{abstract}

\section{Introduction}

Let us  consider a controlled Markov  state process $X$ $=$ $(X^\alpha)_t$ valued in $\Xc$ $\subset$ $\R^d$ with a control process $\alpha$ $=$ $(\alpha_t)$ valued in $A$ $\subset$ $\R^m$.  Given an open set $\Oc$ of $\Xc$, 
we denote by $\tau$ $=$ $\tau^\alpha$ the exit time of the domain $\Oc$ before a terminal horizon 
$T$ $<$ $\infty$, i.e., 
\begin{align}
\tau &= \;  \inf\{ t \geq 0: X_t \notin \Oc \} \wedge T, 
\end{align}
with the usual convention that $\inf\emptyset$ $=$ $\infty$. 
The objective is then to maximize over control process $\alpha$ a criterion in the form
\begin{align} \label{Mayer}
J(\alpha) &= \; \E\big[ g(X_\tau^\alpha) \big], \quad \rightarrow \quad V_0 \; = \; \sup_\alpha J(\alpha), 
\end{align}
for some terminal reward function $g$ on $\R^d$. 
In typical examples,  $X$ is  modelled by a controlled diffusion process as
\begin{align} \label{diffX} 
\d X_t &= \; \mu(X_t,\alpha_t) \d t + \sigma(X_t,\alpha_t) \d W_t, 
\end{align}
and we can also consider jump-diffusion processes, which is  in particular relevant for  
insurance/reinsurance problem with minimization of the ruin probability in finite time.   

\begin{rem}
Notice that there is no loss of generality to focus on the above Mayer form, as the case of Bolza criterion with running reward:
\begin{align}
J(\alpha) &= \; \E\big[ \int_0^\tau f(X_t^\alpha,\alpha_t) \d t +  g(X_\tau^\alpha) \big],
\end{align}
can be reduced to the Mayer form by considering as usual the additional component $(Y_t)_t$ of the state process, driven by
\begin{align}
\d Y_t &= \; f(X_t,\alpha_t) \d t, 
\end{align}
and the corresponding terminal reward function $\tilde g(x,y)$ $=$ $y+g(x)$. 
\end{rem}
The control problem \eqref{Mayer} with exit time can be solved in a model-based setting, e.g. when the coefficients $\mu$, $\sigma$ in \eqref{diffX}, and the analytical form of $g$ are known,  by PDE methods with splitting scheme as described in appendix \ref{PDE Implementation: Splitting scheme}, and eventually by backward SDE methods, see \cite{boumen09}. In the last years, there has been an important literature about the use of deep learning techniques in the numerical resolution of stochastic control problems and PDEs, which have shown success for notably overcoming the curse of dimensionality,  and  we refer the reader to the  recent surveys by \cite{Bec20} and \cite{Ger21}.  However, these methods do not work well for our class of control problem with exit time. Indeed, for example, when trying to apply the global method of  \cite{Han16} 
by approximating the policy by a neural network with parameters 
$\theta$,  the differentiation of the associated gain function would lead to a Dirac function due to the presence of an indicator function related to the exit time, hence the gradient is ill-posed, which  prevents an efficient implementation of the stochastic gradient ascent algorithm.

In this paper, we propose two types of algorithms based on reinforcement learning for estimating the solution to the control problem \eqref{Mayer} in a model-free setting, i.e., without {\it a priori} knowledge of the model coefficients. We develop policy gradient methods  for learning 
approximate optimal control and value function based on samples of state and rewards. A key feature is to consider parametrized randomized policies, notably for overcoming the issue related to exit time in the  policy gradient  representation.  Our first algorithm   learns directly the optimal policy, while the second type of algorithm is of actor-critic nature by learning alternately the policy and the value function. This can be done either in an offline setting with updates rules based on the whole trajectories of the state, or in an online setting with update rules in real-time incrementally. Our algorithms can be viewed as extensions to controlled processes with exit   time of policy gradients in reinforcement learning usually designed for infinite or finite horizon, see \cite{Sutton}.  

The main application that we develop in this paper for stochastic control in the form \eqref{Mayer} concerns the pricing of buyback options in Stock Repurchase Programs (in short SRPs).  Those are defined as transactions initiated by companies to re-buy their proper stocks for various reasons including the raising of the debt-to-equity ratio or the improvement of earnings per share by reducing the number of outstanding shares. SRPs are also an alternative way to distribute the dividends to the shareholders, 
see \cite{MM61}. For more details about SRPs and its regulatory issues and associated tools, the reader can consult this report \cite{IOSC04}.

There exist several mechanisms for SRPs with complex contracts 
involving investment banks, where   the company mandates a bank to repurchase its shares through a derivative product.  
A well-known example  often used by practitioners is Accelerated Share Repurchases (ASRs), where at time $t=0$ the bank borrows a quantity $B$ of shares required by the company from shareholders, and then purchases  progressively from the open market the quantity $B$ to give it back to shareholders. In addition, the bank becomes in a long position of an American option where at some exercise time $\tau$, the company should pay the bank the average price between $0$ and $\tau$ for each share. 

The valuation of ASRs has recently attracted attention in the literature.
Guéant et al. \cite{Gueant15} consider in a discrete time/space model the pricing of ASRs, which leads to a tree based algorithm. 
Jaimungal et al. \cite{jaimungal16} investigate the same problem in continuous time/space setting by additionally taking into consideration temporary and long-term market impact, and  characterize the execution frontier. Guéant et al. \cite{Gueant19} use    deep learning algorithms in the spirit of \cite{bueetal} and \cite{becchejen19} 
for the pricing of ASRs contracts 
and buyback contract called VWAP-minus profit-sharing. In such contract, the exercise time $\tau$ is chosen by the bank once the amount of shares requested by the company is redeemed. In this paper, we consider a buyback contract where the exercise time $\tau$ is entirely characterized by the execution strategy and can not be chosen by any party. We shall call such a  buyback product as Barrier VWAP-minus. Actually, one can show (see Appendix \ref{appenbarrier}) that in absence of market impact, the price of the Barrier VWAP-minus is equal to the price of the VWAP-minus. 

The pricing of barrier VWAP-minus 
leads to a stochastic control formulation as in \eqref{Mayer} where the exit time is  defined as the first stopping time when the controlled inventory exceeds the quantity of shares to be purchased by the bank within a finite time interval. 
We  implement our algorithms to this pricing problem: since they are model-free, they are robust to model misspecifications, and are valid notably for general model for the stock price including market impact and transaction costs.

We first compare  our numerical results with those obtained by PDE methods with splitting scheme 
as detailed in Appendix \ref{PDE Implementation: Splitting scheme}. 
Our validation test  consists in approximating the optimal policy and then computing the price using Monte Carlo: it provides then by definition a lower bound to the true price of the constrained VWAP-minus contract. 
We show that our model-free policy gradient algorithms yield accurate results similar to PDE schemes designed in a specific model-based setting. It is also less costly and more stable than methods performed in \cite{Gueant19} in a model-based setting,  where the control and the stopping time are parametrized by two distinct neural networks. Moreover, it has the advantage to be easily implemented in general factor models including market impact. We illustrate notably the impact of market impact on the optimal trading policies. 



The rest of the paper is structured as follows.  We develop in Section \ref{sec:PG} the policy gradient approach with randomized policies, and present our two types of  algorithms. Section \ref{SRPsOptions} is devoted to the application to valuation of SRP, including the case with market impact and transaction costs, 
with numerical results illustrating the convergence and accuracy of our algorithms,  and comparison with other methods.  

\section{Policy gradient methods} \label{sec:PG}

We consider a time discretization of the stochastic control problem \eqref{Mayer}. Let $\T$ $=$ $\{t_0 = 0 < \ldots < t_i < \ldots < t_N = T\}$  be a subdivision of $[0,T]$ of size $N$ with time steps $\Delta t_i$ $=$ $t_{i+1}-t_i$, $i$ $=$ $0,\ldots,N-1$.  By misuse of notation, we denote by 
$(X_{t_i})_{i \in \llbracket 0,N\rrbracket}$  the Markov decision process (MDP)  arising from the time discretization of the controlled state process $(X_t)_t$, and it is characterized by an initial distribution $p_0$ for $X_{t_0}$, and 
the transition kernel function $p(.|t_i,x_i,a)$ representing the probability of the next state $X_{t_{i+1}}$ given the current state 
$X_{t_i}$ $=$ $x_i$ $\in$ $\Xc$, and an action $a$ $\in$ $A$ at time $t_i$.  
Notice that in a model-free setting, this transition kernel is unknown.

A randomized policy in this discretized time setting is a measurable transition kernel function $\pi$ $:$ $(t_i,x_i)$ $\in$ $\T\times\Xc$ $\mapsto$ $\pi(.|t_i,x_i)$ $\in$ $\Pc(A)$ (the set of probability measures on $A$), and we say that $\alpha$ $=$ $(\alpha_{t_i})_{i\in\llbracket 0,N-1\rrbracket}$ 
is a randomized feedback control generated from the stochastic policy  $\pi$, written as $\alpha$ $\sim$ $\pi$, when $\alpha_{t_i}$ is drawn from $\pi(.|t_i,X_{t_i})$ at any time $t_i$. 

The exit time of the Markov decision process $(X_{t_i})_{i \in \llbracket 0,N\rrbracket}$ is given by 
\begin{align}
\tau &=\;  \inf\{ t_i \in \T: X_{t_i} \notin \Oc\} \wedge t_N,
\end{align}
and the gain functional associated to the Markov decision process with exit time and randomized feedback control $\alpha$ $\sim$ $\pi$  is given by 
\begin{align}
\mrJ(\pi) &=  \E_{\alpha\sim\pi} \big[ g(X_\tau) \big]. 
\end{align}
Here the notation $\E_{\alpha\sim\pi}[.]$ means that the expectation is taken when the Markov decision process $(X_{t_i})$ is controlled by the randomized feedback control $\alpha$ generated from the stochastic policy $\pi$. 

We now consider stochastic policies $\pi$ $=$ $\pi_\theta$ with parameters $\theta$ $\in$ $\R^D$, and which admit densities  with respect to some measure $\nu$ on $A$: 
$\pi_\theta(\d a|t_i,x_i)$ $=$ $\rho_\theta(t_i,x_i,a) \nu(\d a)$, for some parametrized measurable functions $\rho_\theta$ $:$ $\T\times\Xc\times A$ $\rightarrow$ $(0,\infty)$.  
\begin{itemize}
\item when $A$ is a finite space, say $A$ $=$ $\{a_1,\ldots,a_M\}$, 
we take $\nu$ as the counting measure, and choose softmax policies, i.e., 
\begin{align} \label{softmax} 
\rho_\theta(t_i,x_i,a_m) &= \; \frac{ \exp\big(\phi_{\theta_m}(t_i,x_i) \big) } { \sum_{\ell=1}^M \exp\big(\phi_{\theta_\ell}(t_i,x_i) \big)}, \quad m=1,\ldots,M,
\end{align}
where $\phi_{\theta_m}$ are neural networks on $[0,T]\times\R^d$, and $\theta$ $=$ $(\theta_1,\ldots,\theta_M)$ gathers all the parameters of the $M$ neural networks.  In this case, the score function is given by 
\begin{align}
\nabla_{\theta_\ell} \log \rho_\theta(t_i,x_i,a_m) &= \; 
\big(\delta_{m\ell} - \rho_\theta(t_i,x_i,a_\ell) \big) \nabla_{\theta_\ell} \phi_{\theta_\ell}(t_i,x_i). 
\end{align}
\item when $A$ is a continuous space of $\R^m$, we can choose typically  a Gaussian distribution on $\R^m$ for the stochastic policy, with mean parametrized by neural network $\mu_\theta(t,x)$ valued on $A$, and variance a positive definite matrix $\Sigma$ on $\R^{m\times m}$ to encourage exploration, e.g. $\Sigma$ $=$ $\varepsilon I_m$.  In this case, $\nu$ is the Lebesgue measure on $\R^m$,  and the density is 
\begin{align}
\rho_\theta(t_i,x_i,a) &= \;  \frac{1}{(2\pi)^{m/2} {\rm det}(\Sigma)^{1\over 2}} \exp \big( - \frac{1}{2}\big(a - \mu_\theta(t_i,x_i)\big)\trans \Sigma^{-1} \big(a - \mu_\theta(t_i,x_i)\big) \Big).  
\end{align}
In this case, the score function is given by
\begin{align}
\nabla_\theta \log \rho_\theta(t_i,x_i,a) &= \; 
\nabla_\theta\mu_\theta(t_i,x_i)\trans \Sigma^{-1} (a- \mu_\theta(t_i,x_i)). 
\end{align}
\end{itemize}

We then denote, by abuse of notation, $\mrJ(\theta)$ $=$ $\mrJ(\pi_\theta)$, the performance function viewed as a function of the parameter $\theta$ of the  stochastic policy, and the principle of policy gradient method is to maximize over $\theta$ this function by stochastic gradient ascent algorithm. 
In a model-free setting, the purpose is then to derive a suitable expectation representation of the gradient function $\nabla_\theta \mrJ(\theta)$ that does not involve unknown model coefficients and transition kernel $p(.|t,x,a)$ of the state process, but only sample observations of the states $X_{t_i}$, $i$ $=$ $0,\ldots,N$, hence of the exit time $\tau$,  
when taking decisions $\alpha$ $\sim$ $\pi_\theta$, with known chosen family of densities $\rho_\theta$.

\subsection{Policy gradient representation}

Our first main result is to provide a stochastic policy gradient representation for the performance function $\mrJ$ by adapting arguments in the infinite or finite horizon case. 

\begin{thm} \label{thmPG1} 
We have 
\begin{align} \label{gradient1}
\nabla_\theta \mrJ(\theta) &= \;  \E_{\alpha\sim\pi_\theta} \Big[ g(X_\tau) 
\sum_{i=0}^{N-1}  
\nabla_\theta \log \rho_\theta(t_i,X_{t_i},\alpha_{t_i})  1_{t_i <  \tau}  \Big]. 
\end{align}
\end{thm} 
\begin{proof} For a path $(x_0,\ldots,x_N)$ $\in$ $\Xc^{N+1}$, we denote by 
\begin{align}
\iota(x_0,\ldots,x_N) &= \;  \inf\{ i \in \llbracket 0,N\rrbracket: x_i \notin \Oc \} \wedge N, 
\end{align}
so that the exit time of  $(X_{t_i})_{i \in \llbracket 0,N\rrbracket}$ is written as $\tau$ $=$ $t_{\iota(X_{t_0},\ldots,X_{t_N})}$.  Let us then introduce the function  $G$ defined on $\Xc^{N+1}$ by 
$G(x_0,\ldots,x_N) = \; g( x_{\iota(x_0,\ldots,x_N)})$, 
so that  
\begin{align}
\mrJ(\theta) &= \;  \E_{\alpha\sim\pi_\theta}\big[ G(X_{t_0},\ldots,X_{t_N}) \big] \\
& = \; \int_{\Xc^{N+1}} \int_{A^N}  G(x_0,\ldots,x_N) p_0(\d x_0) \prod_{i=0}^{N-1} \pi_\theta(\d a_i|t_i,x_i) p(\d x_{i+1}|t_i,x_i,a_i) \\
& = \;  \int_{\Xc^{N+1}} \int_{A^N} G(\bolx) p_0(\d x_0) \borho_\theta^N(\bolx,\boa)  \prod_{i=0}^{N-1}  p(\d x_{i+1}|t_i,x_i,a_i) \nu(\d a_i), \label{diffJ} 
\end{align}
where we set $\bolx$ $=$ $(x_0,\ldots,x_N)$, $\boa$ $=$ $(a_0,\ldots,a_{N-1})$, and 
\begin{align}
\borho_\theta^N(\bolx,\boa) &= \;  \prod_{i=0}^{N-1}  \rho_{\theta} (t_i,x_i,a_i). 
\end{align}
By using the classical log-likelihood trick: $\nabla_\theta \borho_\theta^N(\bolx,\boa)$ $=$ $\big(\nabla_\theta \log \borho_\theta^N(\bolx,\boa) \big) \borho_\theta^N(\bolx,\boa)$, and noting that 
\begin{align}
\nabla_\theta \log \borho_\theta^N(\bolx,\boa) &= \; \sum_{i=0}^{N-1} \nabla_\theta \log \rho_\theta(t_i,x_i,a_i), 
\end{align}
we deduce by differentiating \eqref{diffJ} that 
\begin{align}
\nabla_\theta \mrJ(\theta) &= \; \int_{\Xc^{N+1}} \int_{A^N} G(\bolx) \nabla_\theta \log \borho_\theta^N(\bolx,\boa) p_0(\d x_0)  \prod_{i=0}^{N-1} \pi_\theta(\d a_i|t_i,x_i) p(\d x_{i+1}|t_i,x_i,a_i) \\
&= \;  \E_{\alpha\sim\pi_\theta} \Big[ G(X_{t_0},\ldots,X_{t_N})  \sum_{i=0}^{N-1}  \nabla_\theta \log \rho_\theta(t_i,X_{t_i},\alpha_{t_i})   \Big]. 
\end{align}
Finally, observe that for any $i$ $\in$ $\llbracket 0,N-1\rrbracket$,  we have
\begin{align}
& \;  \E_{\alpha\sim\pi_\theta} \Big[ G(X_{t_0},\ldots,X_{t_N})  1_{t_i  \geq \tau}    \nabla_\theta \log \rho_\theta(t_i,X_{t_i},\alpha_{t_i})   \Big] \\
= & \;  \E_{\alpha\sim\pi_\theta} \Big[ g(X_\tau)  1_{t_i \geq  \tau}    \nabla_\theta \log \rho_\theta(t_i,X_{t_i},\alpha_{t_i})   \Big]  \\
 = & \;  \E_{\alpha\sim\pi_\theta} \Big[ g(X_\tau)  1_{t_i \geq \tau}   \E_{\alpha\sim\pi_\theta}\big[  \nabla_\theta \log \rho_\theta(t_i,X_{t_i},\alpha_{t_i}) \big| X_{t_i} \big]   \Big] \\
 = & \;  \E_{\alpha\sim\pi_\theta} \Big[ g(X_\tau)  1_{t_i \geq  \tau}  \underbrace{\nabla_{\theta} \Big( \int_{A} \rho_{\theta}(t_i,X_{t_i}^{},a) \nu(\d a) \Big)}_{=\ 0} \Big]  \; = \; 0, \label{trick}
\end{align}
which yields the required result. 
\end{proof}

 \vspace{2mm}

Alternately, we now provide a second representation formula for the gradient of the performance function by exploiting the dynamic programming.  Let us introduce the dynamic version of $\mrJ$. For $i$ $\in$ $\llbracket 0,N\rrbracket$, and $x$ $\in$ $\Xc$, we define the value (performance) function associated to the policy $\pi_\theta$
\begin{align}
V_i^\theta(x) & := \; \E_{\alpha\sim\pi_\theta} \big[ g(X_{\tau_i}) | X_{t_i} = x \big],
\end{align}
where $\tau_i$ $=$ $\inf\{ t_j \in \T, t_j \geq t_i: X_{t_j} \notin \Oc\} \wedge t_N$, so that $\mrJ(\theta)$ $=$ $\E[V_0^\theta(X_0)]$. We notice that $V_N^\theta(x)$ $=$ $g(x)$, for all $x$ $\in$ $\Xc$, and $V_i^\theta(x)$ $=$ $g(x)$, for all $i$ $\in$ $\llbracket 0,N-1\rrbracket$, and  $x$ $\notin$ $\Oc$. Moreover, 
by  the dynamic programming (which is here simply reduced to the law of conditional expectations), we have for $i$ $\in$ $\llbracket 0,N-1\rrbracket$:
\begin{align} \label{dynpro}
V_i^\theta(x) &= \; \E_{\alpha\sim\pi_\theta} \Big[ V_{i+1}^\theta(X_{t_{i+1}})    
| X_{t_i} = x \Big], \quad \mbox{ for } x \in \Oc. 
\end{align}

\begin{thm} \label{thmAC} 
We have 
\begin{align} \label{repAC}
\nabla_\theta \mrJ(\theta) &= \; \E_{\alpha\sim\pi_\theta} \Big[ \sum_{i=0}^{N-1}  V_{i+1}^\theta(X_{t_{i+1}}) \nabla_\theta \log  \rho_\theta(t_i,X_{t_i},\alpha_{t_i}) 1_{t_i < \tau}  \Big].
\end{align} 
\end{thm}
\begin{proof} From \eqref{dynpro}, we have for $(i,x_i)$ $\in$ $\llbracket 0,N-1\rrbracket \times\Oc$
\begin{align}
V_i^\theta(x_i) &=  \; \int_\Xc \int_A  V_{i+1}^\theta(x_{i+1}) \rho_\theta(t_i,x,a) \nu(\d a) p(\d x_{i+1}| t_i,x_i,a).   
\end{align}
By differentiating with respect to $\theta$, and using again the log-likelihood trick, we get
\begin{align}
\nabla_\theta V_i^\theta(x_i) & = \; \int_\Xc \int_A  \nabla_\theta \big[  V_{i+1}^\theta(x_{i+1}) \big]  \rho_\theta(t_i,x_i,a) \nu(\d a) p(\d x_{i+1}| t_i,x_i,a) \\
& \quad \; + \; \int_\Xc \int_A  V_{i+1}^\theta(x_{i+1}) \nabla_\theta [\log \rho_\theta(t_i,x_i,a) ]  \rho_\theta(t_i,x_i,a) \nu(\d a) p(\d x_{i+1}| t_i,x_i,a) \\ 
& = \;  \int_\Oc \int_A  \nabla_\theta \big[  V_{i+1}^\theta(x_{i+1}) \big]  \pi_\theta(\d a|t_i,x_i)  p(\d x_{i+1}| t_i,x_i,a) \\
& \quad + \; \E_{\alpha\sim\pi_\theta} \Big[   V_{i+1}^\theta(X_{t_{i+1}}) \nabla_\theta \log  \rho_\theta(t_i,X_{t_i},\alpha_{t_i})  | X_{t_i} = x_i \Big], \; i \in \llbracket 0,N-1\rrbracket, 
\end{align}
for all $x_i$ $\in$ $\Oc$, by noting that $\nabla_\theta V_{i+1}^\theta(x)$ $=$ $0$ for $x$ $\notin$ $\Oc$, and $ V_{i+1}^\theta(x)$ $=$ $V_{i+1}^\theta(x)$ for $x$ $\in$ $\Oc$. By iterating over $i$, and 
noting that $\nabla_\theta V_N^\theta(.)$ $\equiv$ $0$, we deduce that for all $x_0\in\mathcal{O}$
\begin{align}
\nabla_\theta V_0^\theta(x_0) &= \; \E_{\alpha\sim\pi_\theta} \Big[ \sum_{i=0}^{N-1}  V_{i+1}^\theta(X_{t_{i+1}}) \nabla_\theta \log  \rho_\theta(t_i,X_{t_i},\alpha_{t_i}) \prod_{j=1}^i 1_{X_{t_j}\in\mathcal{O}} \big| X_{t_0} = x_0 \Big] 
\end{align}
Since $\prod_{j=1}^i 1_{X_{t_j}\in\mathcal{O}} = 1_{t_i<\tau}$ and $\nabla_\theta V_0^\theta(.)$ $=$ $0$ on $\Xc\setminus\Oc$, we get the  required representation formula. 
\end{proof}

\begin{rem}
It is known that stochastic gradient policy algorithms suffer from high variance, and a good alternative is 
to use a baseline. For instance, in the representation \eqref{repAC}, we can substract 
to $ V_{i+1}^\theta(X_{t_{i+1}})$ the term $V_i^{\theta}(X_{t_i}^{})$ without biaising the gradient, i.e. 
\begin{align} \label{repACB}
\nabla_\theta \mrJ(\theta) &= \E_{\alpha\sim\pi_\theta} \Big[ \sum_{i=0}^{N-1} \big(  V_{i+1}^\theta(X_{t_{i+1}}) - V_i^{\theta}(X_{t_i}^{}) \big) \nabla_\theta \log  \rho_\theta(t_i,X_{t_i},\alpha_{t_i}) 1_{t_i<\tau}  \Big], 
\end{align} 
by the same trick as in \eqref{trick}:
\begin{align*}
    &\; \mathbb{E}_{\alpha\sim\pi_\theta} \Big[ V_i^{\theta}(X_{t_i}^{}) 
    \nabla_{\theta} \log \rho_{\theta}(t_i,X_{t_i}^{},\alpha_{t_i}) 1_{t_i<\tau} \Big]\\
    = & \; \mathbb{E}_{\alpha\sim\pi_\theta}\Big[ V_i^{\theta}(X_{t_i}^{}) 1_{t_i<\tau}
    \mathbb{E_{\alpha\sim\pi_\theta}}\big[ \nabla_{\theta} \log \rho_{\theta}(t_i,X_{t_i}^{},\alpha_{t_i})) \mid X_{t_i}^{} \big] \Big]\\
    = & \; \mathbb{E}_{\alpha\sim\pi_\theta}\Big[ V_i^{\theta}(X_{t_i}^{}) 1_{t_i<\tau}\underbrace{\nabla_{\theta} \Big( \int_{A} \rho_{\theta}(t_i,X_{t_i}^{},a) \nu(\d a)  \Big)}_{=\ 0} \Big] \; = \; 0. 
\end{align*}
\end{rem}

\subsection{Algorithms}

We now propose policy gradient algorithms which are based on the representation of the previous section. 
They do not require necessarily the knowledge of model coefficients and transition kernel $p(.|t,x,a)$ of the state process, but only sample observations of the states $X_{t_i}$, $i$ $=$ $0,\ldots,N$, when taking decisions 
$\alpha$ according to the chosen family of randomized policies, via e.g. an environment simulator (blackbox), 
hence of the exit time $\tau$. They do neither require the knowledge of the analytical form of the reward function $g$, and instead, we can consider that given an input/observation  of a state $x$,  
the associated output/reward $g(x)$ is evaluated via e.g. a blackbox simulator.

\vspace{2mm}

Our first algorithm (see pseudo-code in Algorithm \ref{algoPG}) 
is based on the gradient representation \eqref{gradient1}.

\begin{algorithm2e}[H] 
\DontPrintSemicolon 
\SetAlgoLined 
\vspace{1mm}
{\bf Input data}: Number of episodes $E$, mini-batch size $K$, 
learning rate  $\eta$ 
for policy gradient estimation;  
Parametrized family of randomized policies $\pi_\theta$ with densities $\rho_\theta$; \\ 
{\bf Initialization}: parameter $\theta$; \\
\For{each episode $e$ $=$ $1,\ldots,E$}
{select a random path $k$ $=$ $1,\ldots,K$; \\
Initialize state $X_0^{(k)}$ $\in$ $\Oc$; \\
\For{$i$ $=$ $0,\ldots,N-1$} 
 {Generate action $\alpha^{(k)}_{t_i}$ $\sim$ $\pi_\theta(.|t_i,X_{t_i}^{(k)})$ \\
Simulate by a model or observe (e.g. by blackbox) state $X_{t_{i+1}}^{(k)}$ \\
If $X_{t_{i+1}}^{(k)}$ $\notin$ $\Oc$ or $t_{i+1}$ $=$ $T$, store the exit time $\tau^{(k)}$ $=$ $t_{i+1}$, 
compute or observe by blackbox $G^{(k)}$ $:=$ $g(X_{\tau^{(k)}}^{(k)})$, and close the loop; \\
Otherwise $i$ $\leftarrow$ $i+1$;
} 
{Compute for path $k$
\begin{align}
\Gamma_\theta^{(k)} := G^{(k)} 
\sum_{t_i < \tau^{(k)}}^{}  
\nabla_\theta \log \rho_\theta(t_i,X_{t_i}^{(k)},\alpha^{(k)}_{t_i})  
\end{align}
Update parameters of the policies: $\theta$ $\leftarrow$ $\theta$ $+$ $\eta$ 
$\frac{1}{K}$ $\sum_{k=1}^K \Gamma_\theta^{(k)}$; 
}
}
{\bf Return}: $\pi_\theta$
\caption{Stochastic gradient policy  \label{algoPG} }
\end{algorithm2e}

\vspace{3mm}

Our second type of algorithm is based on the gradient representation \eqref{repACB}, and is of actor-critic type: it consists in estimating simultaneously via fixed-point iterations the randomized optimal policy (the actor)  by policy gradient (PG), and the value function (critic) by performance evaluation relying on the martingale property relation \eqref{dynpro}.   
More precisely, in addition to the parametrized family $\pi_\theta$ of randomized policies, we are given a family of functions $\Vc_\phi$ on $[0,T]\times\Xc$, with parameter $\phi$, e.g. neural network, aiming to approximate the value function. The parameters $(\theta,\phi)$ are then updated alternately as follows: given a current estimation 
$(\theta^{(n)},\phi^{(n)})$, the parameter $\theta$ is updated according to the PG \eqref{repACB} by replacing $V$ by $\Vc_{\phi^{(n)}}$: 
\begin{align}
\theta^{(n+1)} & = \; \theta^{(n)} + \eta \E_{\alpha\sim\pi_{\theta^{(n)}} }\Big[ \sum_{t_i< \tau} \big(\Vc_{\phi^{(n)}}(t_{i+1},X_{t_{i+1}}) 
- \Vc_{\phi^{(n)}}(t_i,X_{t_i}) \big) 
\nabla_\theta \log \rho_{\theta^{(n)}}(t_i,X^{}_{t_i},\alpha_{t_i}) \Big]
\end{align} 
while $\phi$ is updated by minimizing the square regression error: 
\begin{align}
\E\Big[ \Big| \Vc_{\phi^{(n)}}(t_{i+1},X_{t_{i+1}}) - \Vc_\phi(t_i,X_{t_i}) \Big|^2 1_{X_{t_i} \in \Oc} \Big]. 
\end{align}
Notice that we only need to learn the value function on the domain $\Oc$ by sampling the state process until the exit time $\tau$, as it is extended on $\Xc\setminus\Oc$ by the reward $g$. 

The pseudo-code of our Actor-Critic algorithm is described in  Algorithm  \ref{AlgoAC1}.

\vspace{2mm}

\begin{algorithm2e}[H] 
\DontPrintSemicolon 
\SetAlgoLined 
\vspace{1mm}
{\bf Input data}: Number of episodes $E$, mini-batch size $K$, 
learning rates  $\eta^G$, $\eta^V$ for policy and value function estimation; 
Parametrized family $\pi_\theta$ with densities $\rho_\theta$ for randomized policies, and $\Vc_\phi$ for value function; 
\\
{\bf Initialization}: parameter $\theta$, $\phi$; \\
\For{each episode $e$ $=$ $1,\ldots,E$}
{select a random path $k$ $=$ $1,\ldots,K$; \\
Initialize state $X_0^{(k)}$ $\in$ $\Oc$; \\
\For{$i$ $=$ $0,\ldots,N-1$} 
{
Generate action $\alpha^{(k)}_{t_i}$ $\sim$ $\pi_\theta(.|t_i,X_{t_i}^{(k)})$ \\
Simulate by a model or observe (e.g. by blackbox) state $X_{t_{i+1}}^{(k)}$ \\
If $X_{t_{i+1}}^{(k)}$ $\notin$ $\Oc$ or $t_{i+1}$ $=$ $T$, set $\tau^{(k)}$ $=$ $t_{i+1}$, 
$\Vc_\phi(t_{i+1},X_{t_{i+1}}^{(k)})$ $=$ $g(X_{t_{i+1}}^{(k)})$ computed e.g. by blackbox, 
and close the loop; \\
Otherwise $i$ $\leftarrow$ $i+1$;
} 
{Compute for path $k$
\begin{align}
\Gamma_\theta^{(k)} &:=  \sum_{t_i < \tau^{(k)}}^{}  \big(\Vc_\phi(t_{i+1},X_{t_{i+1}}^{(k)}) - \Vc_\phi(t_i,X^{(k)}_{t_i}) \big) 
\nabla_\theta \log \rho_\theta(t_i,X_{t_i}^{(k)},\alpha^{(k)}_{t_i})  \\
\Delta_\phi^{(k)} &:= \sum_{t_i< \tau^{(k)}} \big(\Vc_\phi(t_{i+1},X_{t_{i+1}}^{(k)}) - \Vc_\phi(t_i,X^{(k)}_{t_i}) \big) 
\nabla_\phi \Vc_\phi (t_i,X^{(k)}_{t_i}) 
\end{align}

Actor update: $\theta$ $\leftarrow$ $\theta$ $+$ $\eta^G$ $\frac{1}{K}$ $\sum_{k=1}^K \Gamma_\theta^{(k)}$; \\
Critic update: $\phi$ $\leftarrow$ $\phi$ $+$ $\eta^V$ $\frac{1}{K}$ $\sum_{k=1}^K \Delta_\phi^{(k)}$; 
}
}
{\bf Return}: $\pi_\theta$, $\Vc_\phi$. 
\caption{Actor-Critic (offline) \label{AlgoAC1} }
\end{algorithm2e}

\vspace{3mm}

In the above actor-critic algorithm, the parameters are updated once the whole state trajectories are sampled. 
We can design an online version where the parameters are updated in real-time incrementally, see pseudo-code in Algorithm \ref{AlgoAC2}.  

\vspace{1mm}

\begin{algorithm2e}[H] 
\DontPrintSemicolon 
\SetAlgoLined 
\vspace{1mm}
{\bf Input data}: Number of episodes $E$, mini-batch size $K$, 
learning rates  $\eta^G$, $\eta^V$ for policy and value function estimation; 
Parametrized family $\pi_\theta$ with densities $\rho_\theta$ for randomized policies, and $\Vc_\phi$ for value function; 
\\
{\bf Initialization}: parameter $\theta$, $\phi$; \\
\For{each episode $e$ $=$ $1,\ldots,E$}
{select a random path $k$ $=$ $1,\ldots,K$; \\
Initialize state $X_0^{(k)}$ $\in$ $\Oc$; \\
\For{$i$ $=$ $0,\ldots,N-1$} 
{
Generate action $\alpha^{(k)}_{t_i}$ $\sim$ $\pi_\theta(.|t_i,X_{t_i}^{(k)})$ \\
Simulate by a model or observe (e.g. by blackbox) state $X_{t_{i+1}}^{(k)}$ \\
If $X_{t_{i+1}}^{(k)}$ $\notin$ $\Oc$ or $t_{i+1}$ $=$ $T$, set $\tau^{(k)}$ $=$ $t_{i+1}$, 
$\Vc_\phi(t_{i+1},X_{t_{i+1}}^{(k)})$ $=$ $g(X_{t_{i+1}}^{(k)})$ computed e.g. by blackbox\\
Actor update: $$\theta \leftarrow \theta + \eta^\mathrm{G}\big(\Vc_\phi(t_{i+1},X_{t_{i+1}}^{(k)}) - \Vc_\phi(t_i,X^{(k)}_{t_i}) \big) \nabla_\theta \log \rho_\theta(t_i,X_{t_i}^{(k)},\alpha^{(k)}_{t_i})$$\\
Critic update: 
$$\phi \leftarrow \phi + \eta^\mathrm{V} \big(\Vc_\phi(t_{i+1},X_{t_{i+1}}^{(k)}) - \Vc_\phi(t_i,X^{(k)}_{t_i}) \big) 
\nabla_\phi \Vc_\phi (t_i,X^{(k)}_{t_i}) $$\\
If $X_{t_{i+1}}^{(k)}$ $\notin$ $\Oc$ or $t_{i+1}$ $=$ $T$, close the loop; Otherwise $i$ $\leftarrow$ $i+1$;
} 
}
{\bf Return}: $\pi_\theta$, $\Vc_\phi$. 
\caption{Actor-Critic (online) \label{AlgoAC2} }
\end{algorithm2e}

\section{Application to Share Repurchase Programs Pricing} \label{SRPsOptions}

\subsection{Problem formulation} \label{subsec:notations}

We consider a company/client with stock price $S$.  This client mandates a bank to buy a quantity $B$ of shares of stock within a period $[0,T]$. At early termination date $\tau$ or at maturity $T$ if no early termination has appeared, the client pays to the bank the 
Volume Weighted Average Price (in short VWAP) defined as $V_\tau:={1 \over \tau}\int_0^\tau S_t dt$,  discounted by the number of shares, i.e.,  the amount $B \; V_{ \tau}$. The bank gives to the client the quantity $B$ of shares, and its value at $\tau$ is $BS_{\tau}$. From the bank perspective, it is equivalent to being long an option with payoff $B(V_{\tau}- S_{\tau})$ at $\tau$. If the bank fails to collect the quantity $B$ before $T$ for the company, it must pay a penalty to the client. 
For the sake of simplicity, we have  not included rate, dividends and repo, although this can be easily incorporated.

We denote by $(Q_t)_{t\in[0,T]}$ the quantity of shares (inventory) hold by the trader of the bank, and 
governed by 
\begin{equation}
    dQ_{t}=\alpha_{t}dt,
\end{equation}
where $\alpha$ represents the trading speed, valued in $[0,\overline{a}]$, for some constant 
$\overline{a}$ $\in$ $(0,\infty)$. 
The underlying stock price $S$ is a continuous time process, possibly controlled by $\alpha$ in presence of  permanent market impact. The dynamics of the VWAP process $(V_t)_{t}$ and of the cumulated cost process $(C_t)_t$ are given by 
\begin{align}
\d V_t & = \; \Big( \frac{S_t - V_t}{t} \Big) \d t, \quad 0 < t \leq T, \; V_0 \; = \: S_0, \quad \d C_t \; = \; \alpha_t S_t \d t, \; C_0 = 0. 
\end{align}
The profit and loss (PnL) of the bank at execution time $\tau$ $\leq$ $T$ is then given by
\begin{align}
    \mathrm{PnL}_\tau^\alpha  = B(V_\tau - S_\tau) - \lambda(B-Q_\tau)_+ - \beta BC_\tau, 
\end{align}
where $\lambda>0$ is a penalization parameter, effective when $\tau$ $=$ $T$, and $Q_T < B$, and $\beta$ $\geq$ $0$ is a transaction cost parameter. 
The price of the barrier VWAP-minus contract is determined by the following stochastic control problem
\begin{equation} \label{SCP}
    P_{{BV}} := \sup_{\alpha\in\Ac} \; 
    \E \big[  \mathrm{PnL}_{\tau^\alpha}^{\alpha} \big], 
\end{equation}
where $\mathcal{A}$ is the set of admissible trading strategies, and $\tau^\alpha$  
$:=\inf\{t>0\mid Q_t \geq B \} \wedge T$ is the early termination time of the contract, defined as the first time when the inventory exceeds the required quantity $B$ of shares.
This fits into the form \eqref{Mayer} with state variables $X$ $=$ $(S,V,Q,C)$. 

\begin{rem}
In this context, the price of the ASR is given by 
\begin{align}
P_{ASR} & := \;  \sup_{\alpha\in\Ac} \sup_{\bar{\tau}\in\Tc_{0,T}}  \E \big[  \mathrm{PnL}_{\bar{\tau}}^{\alpha} \big],
\end{align}
while the price of the VWAP-minus contract as considered in \cite{Gueant19} is given by 
\begin{align}
P_{V} & := \; \sup_{\alpha\in\Ac}  \sup_{\bar{\tau}\in\Tc_{\tau^\alpha,T}} \E \big[  \mathrm{PnL}_{\bar{\tau}}^{\alpha} \big],
\end{align}
where $\Tc_{t,T}$ is the set of stopping times valued in $[t,T]$. The prices of these contracts have been computed in \cite{Gueant19} by using two distinct neural networks for approximating the policy $\alpha$ and the stopping time $\bar\tau$, and by definition, we should have $P_{ASR}$ $\geq$ $P_V$ $\geq$ $P_{BV}$. Actually, one can show  that $P_V$ $=$ $P_{BV}$  in absence of market impact and transaction costs, see Appendix \ref{appenbarrier}. In other words, the pricing problem for the VWAP-minus can be reduced to a stochastic control with exit time, and there is no need to consider an additional optimization over stopping times $\bar\tau$, which is quite advantageous from a numerical point of view.  
\end{rem}

The algorithm proposed in \cite{Gueant19} considers two neural networks: $p_{\theta}$ for the randomized stopping time and $\mra_{\xi}$ for trading rate  to estimate the optimal strategy leading to $P_V$. The optimisation is performed by a stochastic gradient ascent with the loss function
\begin{equation}
    \mathcal{L}(\theta,\xi) = \mathbb{E}\Big[ \sum_{i=0}^{N-1} \prod_{j=0}^{i-1} \left(1-p_{\theta}(t_j,X_{t_j})\right)p_{\theta}(t_i,X_{t_i})\mathrm{PnL_{t_i}} + \prod_{j=0}^{N-1}\left(1-p_{\theta}(t_j,X_{t_j})\right)\mathrm{PnL_{t_N}} \Big]. 
\end{equation}
Here $\prod_{j=0}^{i-1} \left(1-p_{\theta}(t_j,X_{t_j})\right)p_{\theta}(t_i,X_{t_i})$ represents the probability to exercise at $t_i$, for a given path of the state variables. For the profit and loss $\mathrm{PnL}$, $(B-Q_{t_i})^+$ is replaced by $|B-Q_{t_i}|$ to prevent  the agent from buying once the barrier is reached. Notice that the computation of the gradient of $\mathcal{L}$ with respect to $\theta$ and $\xi$ is extremely costly. Furthermore, the numerical experiments show  highly unstable results. Instead, our policy gradient algorithms is less costly and show stable results.

\subsection{Numerical results}

For the numerical results and comparison with other methods, we consider a price process with 
linear permanent price impact, governed by 
\begin{align}
\d S_t & = \; S_t \big( \gamma \alpha_t \d t + \sigma \d W_t), \quad 0 \leq t \leq T,
\end{align}
where $\gamma$ $\geq$ $0$ is a constant market impact parameter.  
The value function $P(t,x)$ with $t$ $\in$ $[0,T]$, $x$ $=$ $(s,v,q,c)$ $\in$ $\R_+^*\times\R_+^*\times\R_+\times\R_+$, is solution to the Bellman equation: 
\begin{align}
& \partial_t P  + 
\overline{a} \big( \gamma s\partial_s P + s \partial_c P + \partial_q P\big)^+  \label{PDE3d} \\
 + \;    \frac{s-v}{t} \partial_v P + \frac{1}{2} \sigma^2 s^2 \partial_s^2 P & \; = 0, \quad  t \in (0,T), 
(s,v,q,c) \in \R_+^*\times\R_+^*\times [0,B)\times\R_+, \nonumber
\end{align}
with the boundary conditions: 
\begin{equation}
\left\{
\begin{array}{ccl}
P(t,x) &=& B(v-s) - \beta B c, \quad t \in [0,T], (s,v,q,c) \in \R_+^*\times\R_+^*\times [B,\infty)\times\R_+, \\
P(T,x) &=& B(v-s) - \lambda(B-q)_+ - \beta B c, \quad (s,v,q,c) \in \R_+^*\times\R_+^*\times \R_+ \times\R_+.  
\end{array}
\right. 
\end{equation}
Notice that the optimal feedback control is of bang-bang type, namely:
\begin{align}
\hat a(t,x) &=\; 
\left\{
\begin{array}{cl}
0 & \mbox{ if } \; \gamma s\partial_s P + s \partial_c P + \partial_q P \leq 0, \\
\overline{a} & \mbox{ otherwise},
\end{array}
\right.
\end{align}
and therefore, we shall consider a softmax randomized policy as in \eqref{softmax} with  two possible values in 
$\{0,\overline{a}\}$.

For numerical experiments of our algorithms to the pricing of Barrier VWAP-minus, we neglect transaction costs $\beta$ $=$ $0$, and take the following parameters:  $T=60$ days, $S_0=1$, $B=1$, and $\overline{a}$ ranging from $5.04$ to $25.2$, 
$\lambda=5$, $\Delta t=1/252$, 
number of Monte-Carlo simulations: $N_\mathrm{MC}=10^5$. 

For the architecture of the neural networks for the randomized policies and the value function (for the actor-critic AC algorithm), 
we have used neural networks with $2$ hidden layers of dimension $8$ (linear ouput and Relu as intermediate activation function). The SGD is an Adam algorithm with standard hyper-parameters and $64$ as mini-batch size for SGP and $32$ for AC\footnote{The algorithm has been written from scratch in $C_{++}$.}. 
We first compute the price $P_{BV} \times 10^4$ in absence of market impact $\gamma$ $=$ $0$, and compare with the results obtained by HJB solver\footnote{We thank A. Conze and J. Adrien for their contributions to the PDE implementation of this project.} (see Appendix \ref{PDE Implementation: Splitting scheme}). 
We fix $\sigma$ $=$ $0.2$, and vary the maximal trading rate $\overline{a}$, and display the associated prices in Figure \ref{fig1}. By construction, as we compute the expectation for a sub-optimal control, we obtain a lower bound. In particular, as the underlying price process is a martingale, note that using a constant control, we get $\bf 0$ bp. 
The graph of convergence in terms of the number of episodes of the algorithm  for two pairs of parameters  of $(\overline{a},\sigma)$, 
is reported in Figure \ref{fig2}.

\begin{figure}[H]
\begin{center}
\includegraphics[width=10cm,height=6cm]{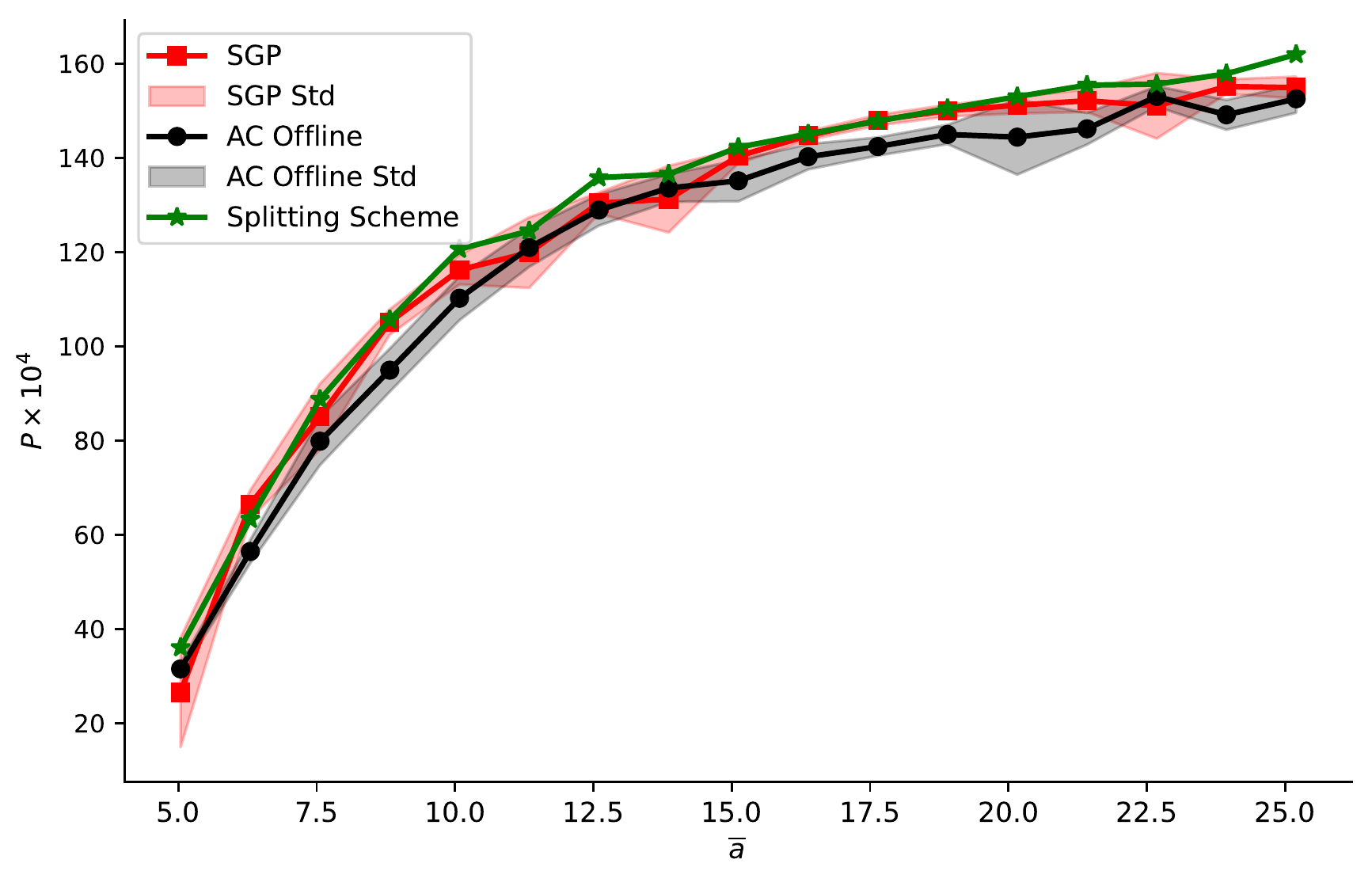}
\end{center}
\caption{ $P_{BV} \times 10^4$ in absence of market impact and transaction costs for different values of $\overline{a}$ computed with stochastic gradient policy and actor critic compared to splitting scheme (HJB solver). 
}
\label{fig1}
\end{figure}

\begin{figure}[H]
\begin{center}
\includegraphics[width=5cm,height=5cm]{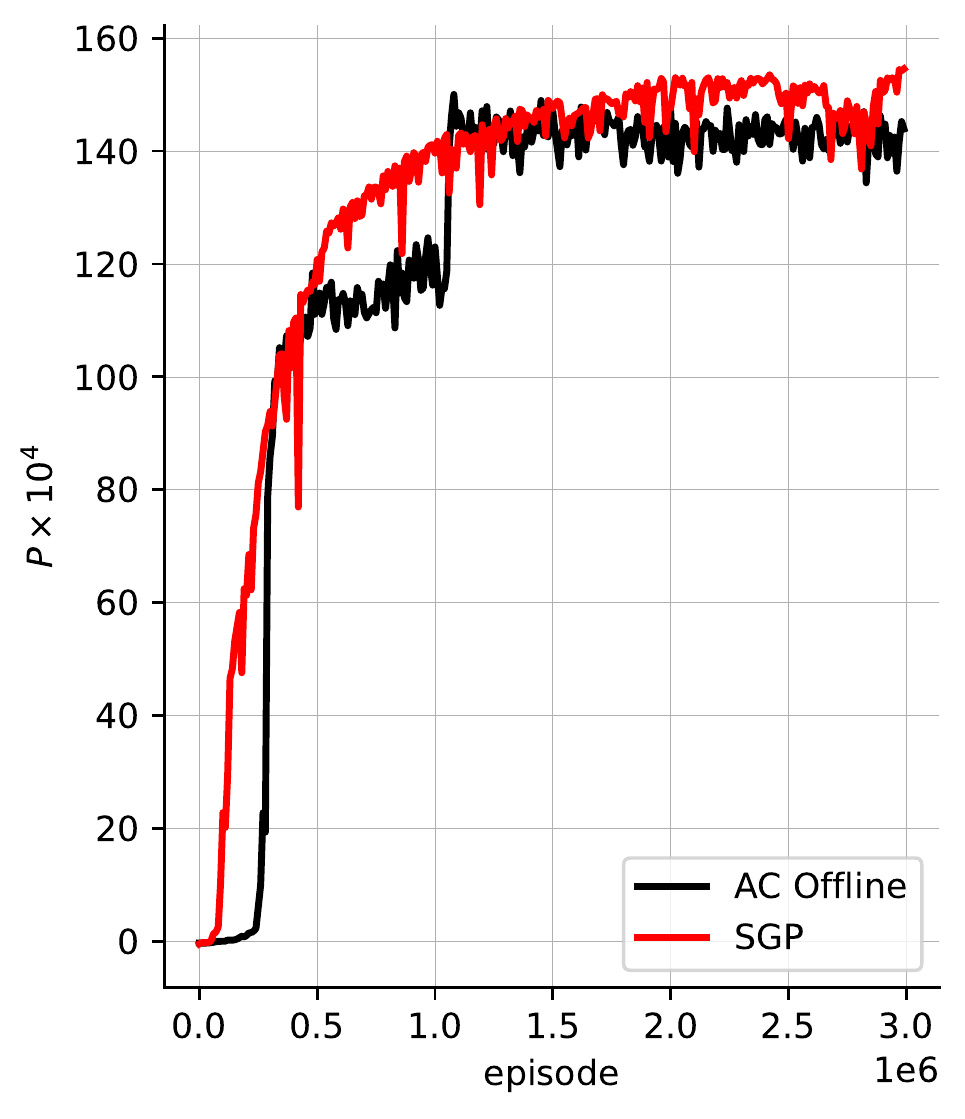}
\includegraphics[width=5cm,height=5cm]{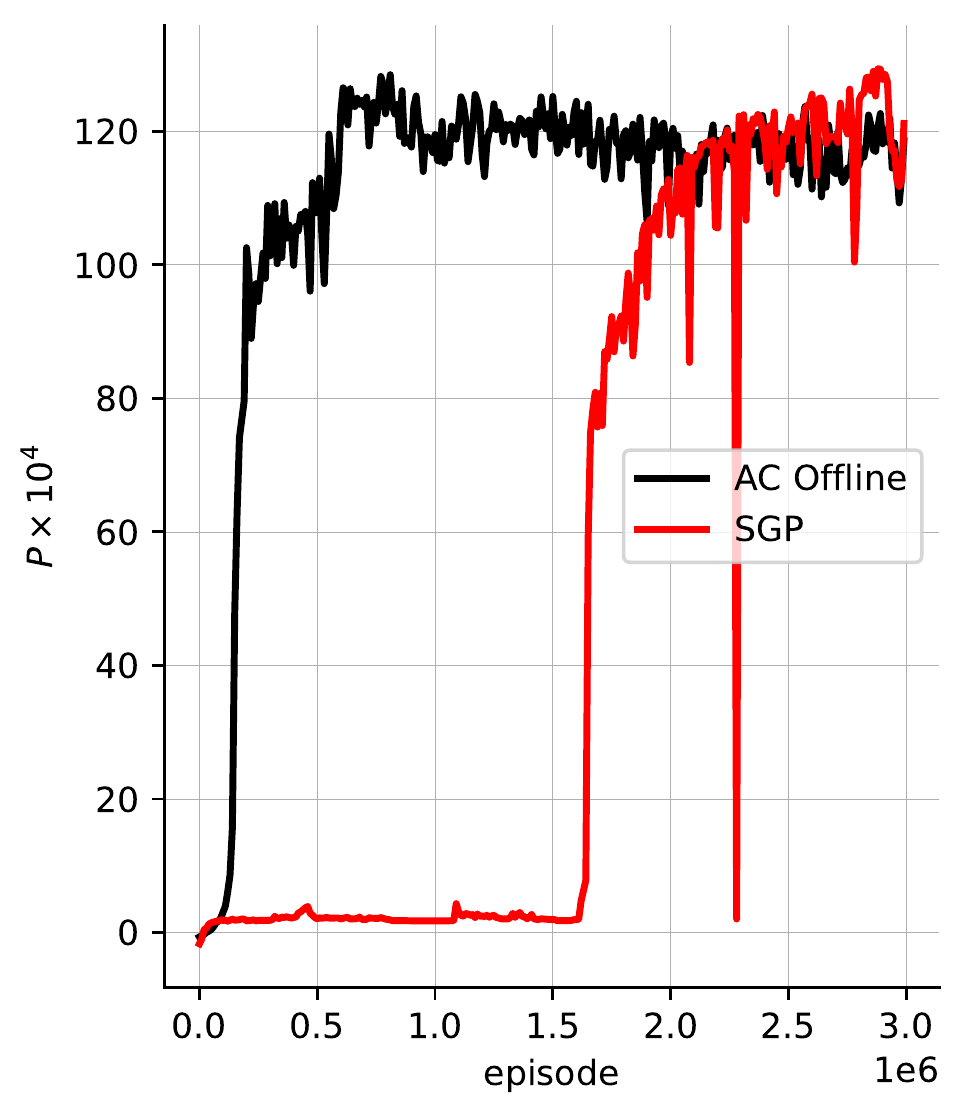}
\end{center}
\caption{Convergence as a function of iterations for $P_{BV} \times 10^4$ (without market impact and transaction costs) for $\overline{a}=36.5, \sigma=0.2$ (left) and $\overline{a}=9, \sigma=0.25$ (right)}
\label{fig2}
\end{figure}

The two algorithms (SGP and AC) produce results that are similar to those obtained using splitting scheme in terms of price. Furthermore, the execution time of these algorithms is also found to be comparable to that of HJB solver, with both methods taking about two minutes to converge, indicating that they are computationally efficient and capable of solving the problem in a timely manner. However, when the number of state variables increases, the PDE method becomes computationally very costly in comparison to our proposed methods. This means that for problems involving a large number of state variables, our method becomes the only viable option. Overall, the results of this study demonstrate that our proposed algorithms are a reliable and cost-effective alternative to the PDE method for solving this class of problems.

Next, we display the surface of the optimal randomized policy for fixed spot price $S$, for two different values of $t$ ($t$ $=$ $T/2$ and $t$ near maturity $T$), and as a function of the VWAP and inventory. Figure \ref{fig3a} shows the results in absence of market impact while Figure \ref{fig3b} considers the case with market impact. 
We observe that when we are close to the maturity, the probability of choosing the maximal trading rate is equal to one for almost all states of the VWAP and inventory with or without market impact: this is due to the fact that the trader has to achieve the goal of repurchasing the requested quantity of shares as he would be penalized otherwise.  
When we are in the midterm of the program, the  optimal policy  consists in choosing the maximal trading rate only when the VWAP is larger than some threshold, say $V^*$, 
as he has enough time to complete his repurchasing goal. 
In absence of market impact, this threshold $V^*$ is approximately equal to the spot price, while in presence of market impact,  
this threshold decreases with the market impact and also with the inventory. In other words, the trader will buy more quickly some fraction of the total shares $B$ as the market impact is more penalizing when approaching maturity.

\begin{figure}[H]
\begin{center}
\includegraphics[width=13cm,height=5.5cm]{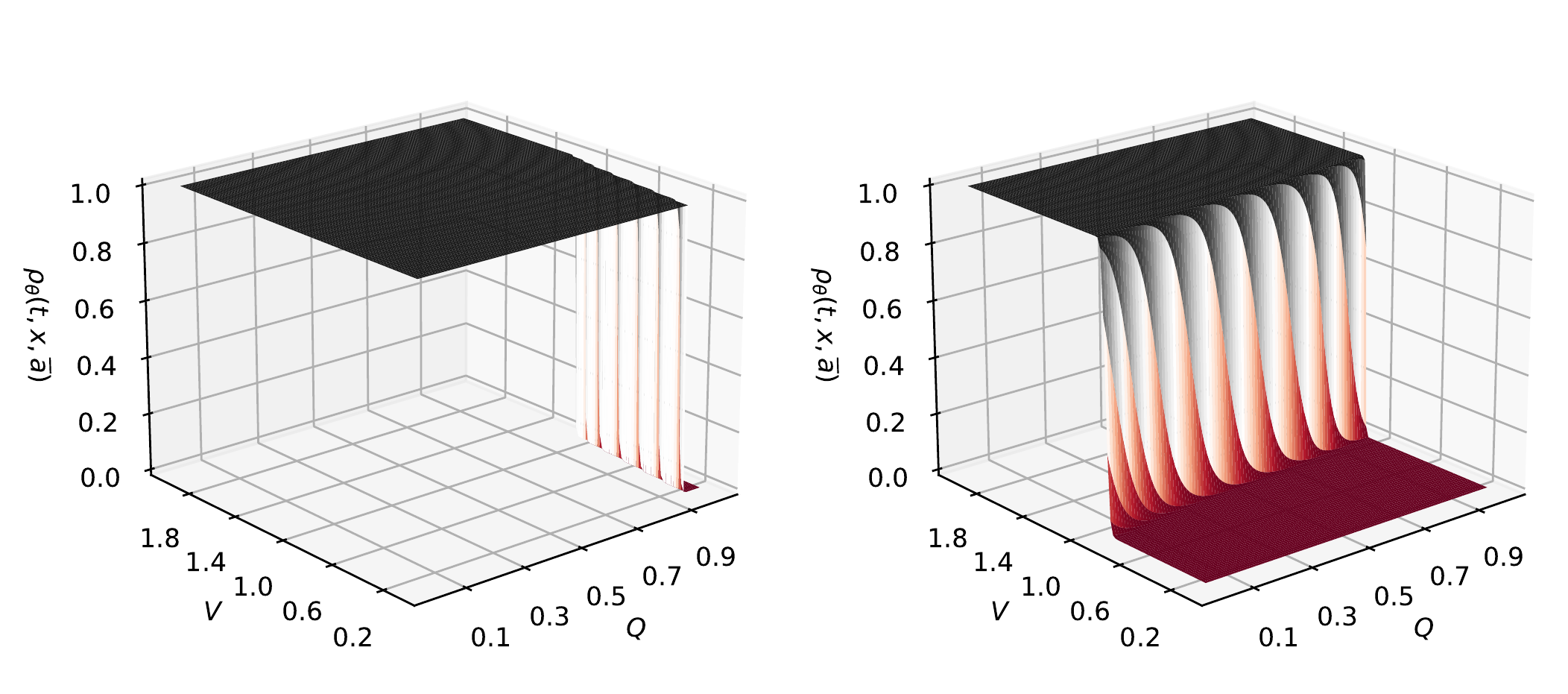}
\end{center}
\caption{Optimal policy $\rho_{\theta}(t,x,\overline{a})$ in absence of market impact and transaction costs for $Q\in[0,1], V\in[0.1,2]$, $S=1$, $\sigma$ $=$ $0.2$, $t=T-dt$ (left) and $t=\frac{T}{2}$ (right).}
\label{fig3a}
\end{figure}

\begin{figure}[H]
\begin{center}
\includegraphics[width=13cm,height=5.5cm]{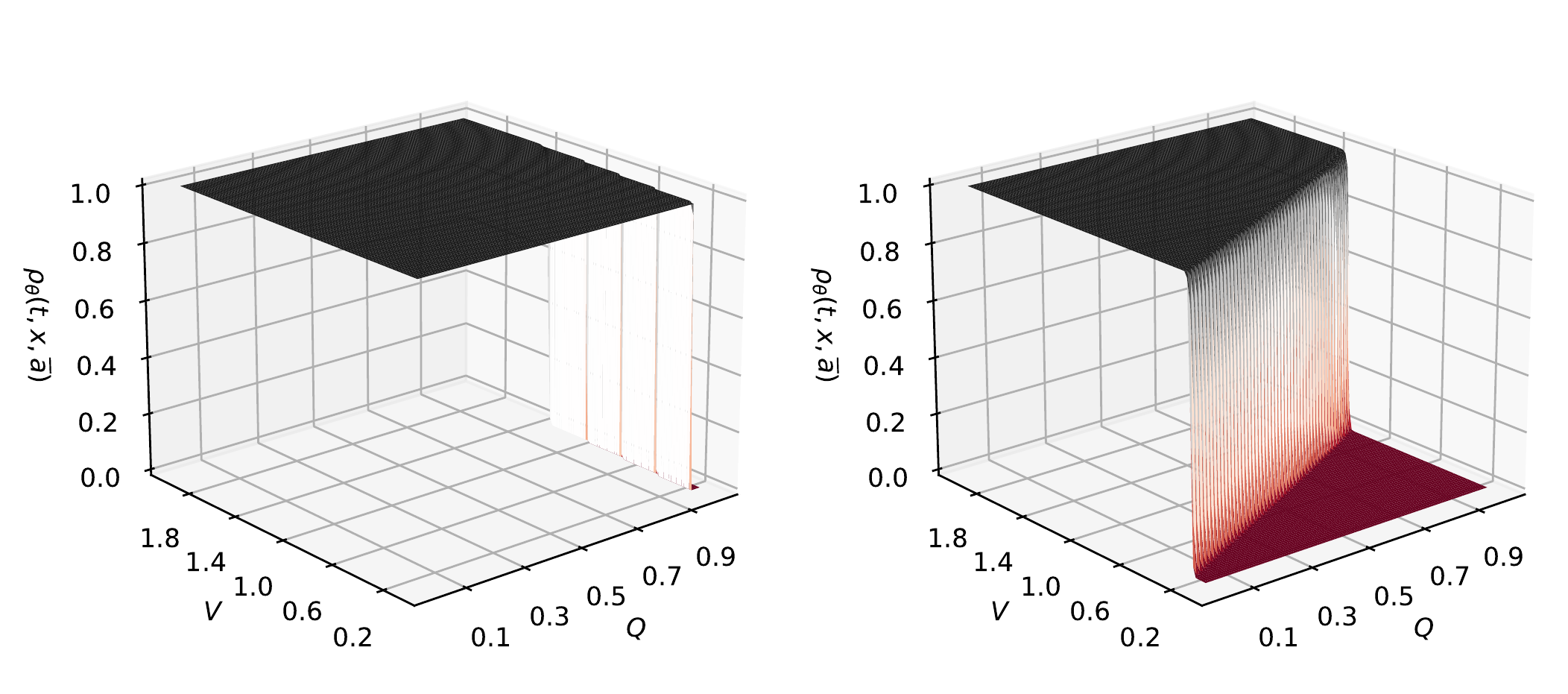}
\end{center}
\caption{Optimal policy $\rho_{\theta}(t,x,\overline{a})$ when market impact is included ($\gamma=0.1$) for $Q\in[0,1], V\in[0.1,2]$, $S=1$, $\sigma$ $=$ $0.2$, $t=T-dt$ (left) and $t=\frac{T}{2}$ (right).}
\label{fig3b}
\end{figure}

Finally, we represent the evolution of the optimal inventory for two price realizations, in the case without market impact (see Figure \ref{fig4a}) and with market impact (see Figure \ref{fig4b})  
The trader starts by purchasing some fraction of the total shares $B$ (and this is done more quickly and with a higher fraction in presence of market impact), then do not trade for a while until the time when the spot price falls below the VWAP, where he purchases the remaining shared to complete the  buy-back programme.

\begin{figure}[H]
\begin{center}
\includegraphics[width=13cm,height=6cm]{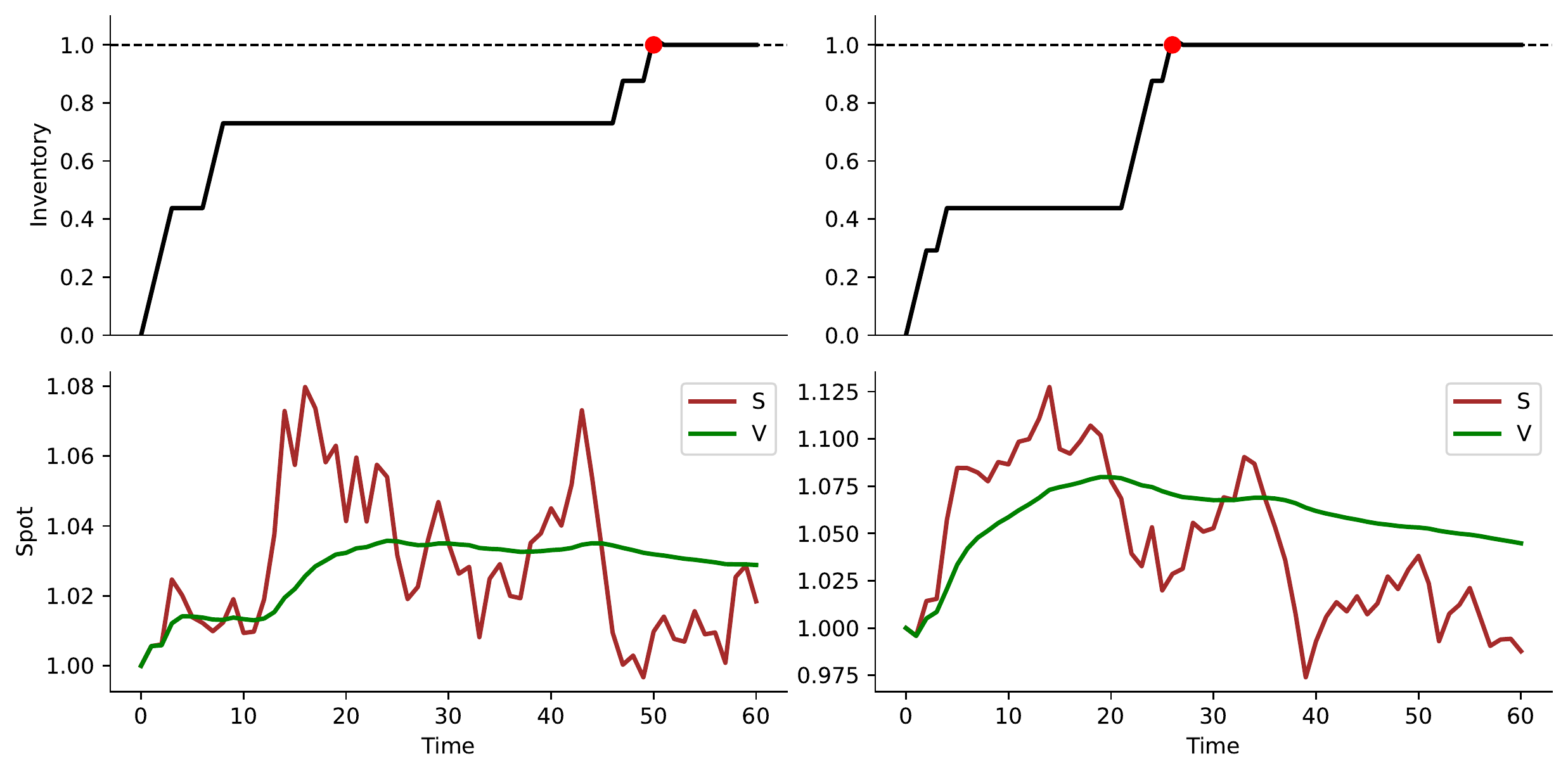}
\end{center}
\caption{Optimal repurchase strategy evolution for two price realizations ($\sigma$ $=$ $0.2$) in absence of market impact and transaction costs.}
\label{fig4a}
\end{figure}
\begin{figure}[H]
\begin{center}
\includegraphics[width=13cm,height=6cm]{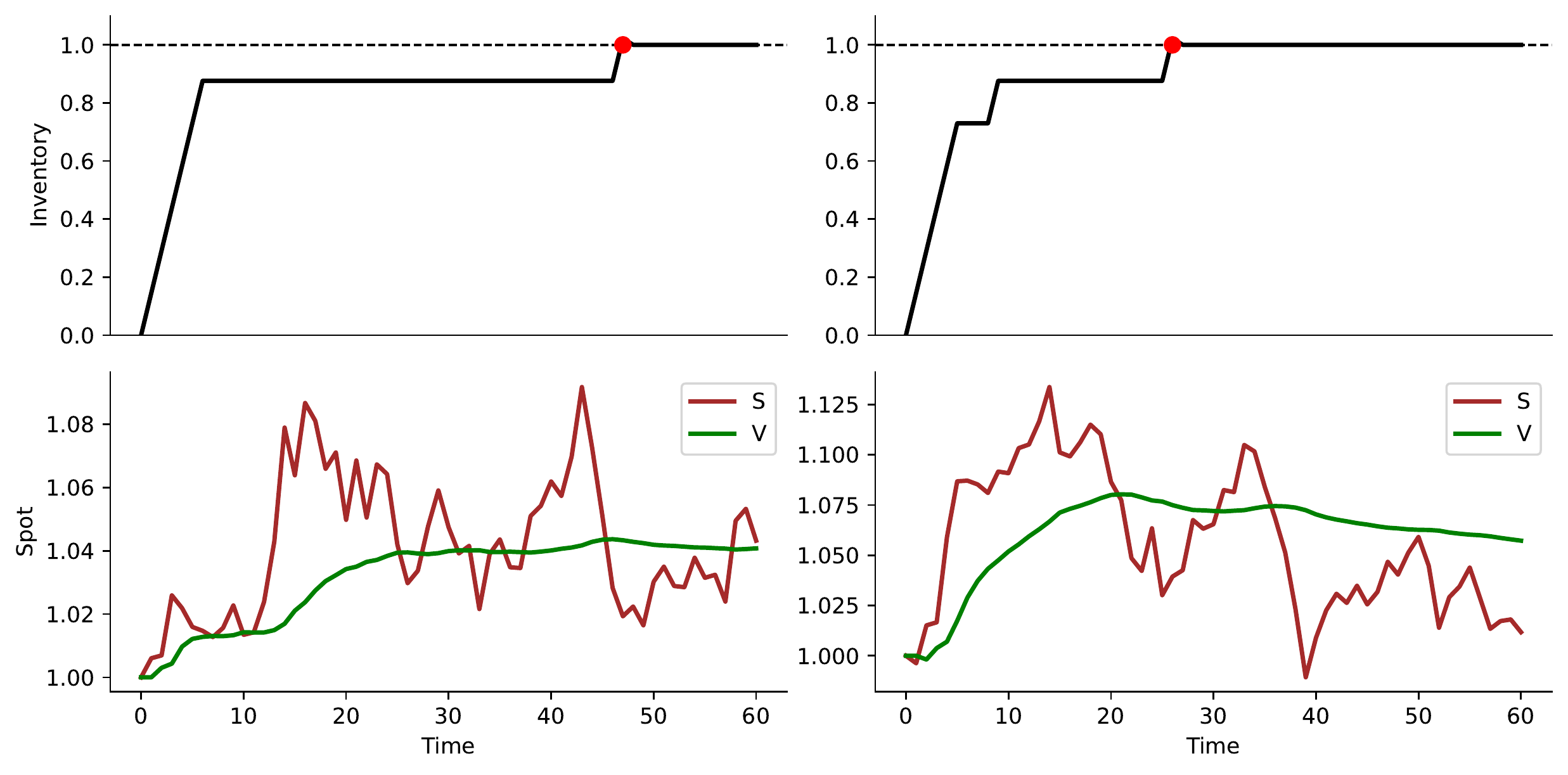}
\end{center}
\caption{Optimal repurchase strategy evolution for two price realizations ($\sigma$ $=$ $0.2$)  with market impact ($\gamma = 0.1$)}
\label{fig4b}
\end{figure}

\appendix

\section{Barrier VWAP-minus vs VWAP-minus} \label{appenbarrier}

Given a trading strategy $\alpha$ $\in$ $\Ac$, valued in $A$ $=$  $[0,\overline{a}]$, we denote by $\tau^\alpha$ the first time when the inventory $Q_t^\alpha$ $=$ $\int_0^t \alpha_s ds$ reaches $B$, and we consider the price of the VWAP-minus and Barrier VWAP-minus given by 
\begin{align}
P_V &= \;  \sup_{\alpha\in\Ac}\sup_{\bar\tau\in\Tc_{\tau^\alpha,T}}\E\big[\mathrm{PNL}^\alpha_{\bar\tau}\big] \quad 
P_{BV} \; = \; \sup_{\alpha\in\Ac} \E\big[\mathrm{PNL}^\alpha_{\tau^\alpha}\big],
\end{align}
where the $\mathrm{PNL}$, in absence of transaction costs, is given by
\begin{align}
\mathrm{PNL}^\alpha_t &= \; B\Big(\frac{1}{t} \int_0^t S_s ds - S_t\Big) - \lambda(B-Q_t^\alpha)_+, 
\quad 0 \leq t\leq T. 
\end{align}
The price process $S$ is a general continuous semimartingale process without market impact, and satisfying 
\begin{align} \label{condS}
\E \big[ \max_{t\in [0,T]} |S_t| \big] & < \;  \infty.
\end{align}
Notice that by Doob's inequality,  such condition \eqref{condS} is satisfied whenever the drift and the volatility of the asset price $S$ are bounded.

\begin{prop}
Under \eqref{condS}, and in absence of market impact and tran\-saction costs, we have $P_{BV}=P_V$. 
\end{prop}
\begin{proof}
Fix some arbitrary $\alpha$ $\in$ $\Ac$, and $\bar\tau$ $\in$ $\Tc_{\tau^\alpha,T}$. 
For $\varepsilon$ $>$ $0$, denote by $\tau_\varepsilon^\alpha$ $=$ $\inf\{ t \geq 0: Q_t^\alpha = B-\varepsilon\} \wedge T$, which is smaller than $\tau^\alpha$, and converges a.s. to $\tau^\alpha$ when $\varepsilon$ goes to zero. Let us then define trading strategy $\alpha^\varepsilon$ $\in$ $\Ac$ by 
\begin{align}
\alpha_t^\varepsilon &=  \;
\left\{
\begin{array}{cl}
\alpha_t & \mbox{ for } \; 0 \leq t \leq  \tau_\varepsilon^\alpha \\
0 & \mbox{ for } \; \tau_\varepsilon^\alpha < t \leq  \bar\tau  \\
\overline{a} & \mbox{ for } \; \bar\tau < t \leq T,
\end{array}
\right. 
\end{align} 
which leads to an associated inventory $Q^{\alpha^\varepsilon}$ given by 
\begin{align}
Q_t^{\alpha^\varepsilon}  &=  \;
\left\{
\begin{array}{cl}
Q_t^{\alpha} & \mbox{ for } \; 0 \leq t \leq  \tau_\varepsilon^\alpha \\
B- \varepsilon & \mbox{ for } \; \tau_\varepsilon^\alpha < t \leq  \bar\tau  \\
B- \varepsilon + \overline{a}(t-\bar\tau) & \mbox{ for } \; \bar\tau < t \leq T. 
\end{array}
\right. 
\end{align}
Notice that $\tau^{\alpha^\varepsilon}$ (the first time when $Q^{\alpha^\varepsilon}$ reaches $B$) is lower-bounded by $\bar\tau$, decreases with $\varepsilon$, and converges a.s. to $\bar\tau$ when $\varepsilon$ goes to zero.

By definition, we have $P_{BV}$ $\geq$ $\E\big[\mathrm{PNL}^{\alpha^\varepsilon}_{\tau^{\alpha^\varepsilon}} \big]$. Let us check that $\mathrm{PNL}^{\alpha^\varepsilon}_{\tau^{\alpha^\varepsilon}}$ converges a.s. to 
$\mathrm{PNL}^{\alpha}_{\bar\tau}$ when $\varepsilon$ goes to zero.  We distinguish two cases:
\begin{itemize}
\item If $\tau^\alpha$ $<$ $T$. Then, $Q^\alpha_{\tau^\alpha}$ $=$ $B$ $\leq$ $Q_{\bar\tau}^\alpha$, and $Q^{\alpha^\varepsilon}_{\tau^{\alpha^\varepsilon}}$ $=$ $B-\varepsilon$ + 
$\overline{a}(\tau^{\alpha^\varepsilon}-\bar\tau)$ converges to $B$ when $\varepsilon$ goes to zero. It follows that 
\begin{align}
\mathrm{PNL}^{\alpha^\varepsilon}_{\tau^{\alpha^\varepsilon}} &= \; 
B\Big(\frac{1}{\tau^{\alpha^\varepsilon}} \int_0^{\tau^{\alpha^\varepsilon}} S_s ds - S_{\tau^{\alpha^\varepsilon}}\Big) 
- \lambda(B-Q_{\tau^{\alpha^\varepsilon}}^{\alpha^\varepsilon})_+ \\
& \rightarrow \; B\Big(\frac{1}{\bar\tau} \int_0^{\bar\tau} S_s ds - S_{\bar\tau}\Big) \; = \; 
\mathrm{PNL}^{\alpha}_{\bar\tau},
\end{align}
as $\varepsilon$ goes to zero. 
\item If $\tau^\alpha$ $=$ $T$. Then $\bar\tau$ $=$ $T$ $=$ $\tau^{\alpha^\varepsilon}$, and $\alpha_t^\varepsilon$ converges to $\alpha_t$, for $0\leq t<T$, when  $\varepsilon$ goes to zero. It follows that $Q_T^{\alpha^\varepsilon}$ 
converges to $Q_T^\alpha$. Therefore, 
\begin{align}
\mathrm{PNL}^{\alpha^\varepsilon}_{\tau^{\alpha^\varepsilon}} &= \; 
B\Big(\frac{1}{T} \int_0^T S_s ds - S_T\Big) - \lambda(B-Q_T^{\alpha^\varepsilon})_+ \\
& \rightarrow \; B\Big(\frac{1}{T} \int_0^T S_s ds - S_T\Big) - \lambda(B-Q_T^{\alpha})_+ \; = \; 
\mathrm{PNL}^{\alpha}_{\bar\tau},
\end{align}
as $\varepsilon$ goes to zero. 
\end{itemize}
Moreover, by noting that $\big| \mathrm{PnL}_{\tau^{\alpha^\varepsilon}}^{\alpha^\varepsilon} \big|$ $\leq$ 
$B( 2\max_{t\in[0,T]} |S_t| + \lambda)$, and under \eqref{condS}, we can apply dominated convergence theorem to deduce that 
\begin{align}
\E\big[\mathrm{PNL}^{\alpha^\varepsilon}_{\tau^{\alpha^\varepsilon}} \big] & \rightarrow  \; 
                \E\big[\mathrm{PNL}^{\alpha}_{\bar\tau} \big], \quad \mbox{ when } \; \varepsilon \mbox{ goes to zero,}
\end{align}
and so $P_{BV}$ $\geq$ $\E\big[\mathrm{PNL}^{\alpha}_{\bar\tau} \big]$. 
Since this holds true for any $\alpha$ $\in$ $\Ac$, and $\bar\tau$ $\in$ $\Tc_{\tau^\alpha,T}$, we conclude that 
$P_{BV}$ $\geq$ $P_V$, hence the equality since it is clear that $P_V$ $\geq$ $P_{BV}$.   
\end{proof}

\section{PDE Implementation by  splitting scheme}  \label{PDE Implementation: Splitting scheme}

\no We solve the Bellman (HJB) equation \eqref{PDE3d} by backward induction. We know $P$ at $T$ 
(\textit{Terminal condition}). Now, we assume that we know $P$ at $t$ and we want to compute $P$ at a previous date $t-\Delta t$. We use the approximation:
\begin{equation}
  \overline{a}  1_{\left\{ (\gamma s\partial_s + s \partial_c + \partial_q)P(t-\Delta t,x) \geq 0 \right\}} \approx
  \overline{a}  1_{\left\{ (\gamma s\partial_s + s \partial_c + \partial_q)P(t,x) \geq 0 \right\}} := \tilde{a}^*(t,x)
\end{equation}
for all $x$ $=$ $(s,v,q,c)$ $\in$ $\mathcal{O}$ $=$ $\R_+^*\times\R_+^*\times(0,B)\times\R_+$. The HJB equation becomes
\begin{align} \label{linearHJBEqua}
    \partial_t P _{\mid_{\cal O}} +{\cal L} P_{\mid_{\cal O}} + {\cal D} P_{\mid_{\cal O}} = 0
\end{align}
where $P_{\mid_{\cal O}}$ is the restriction of $P$ to $\cal O$, $\mathcal{L}$ is a diffusion operator and $\mathcal{D}$ is a transport operator defined over $\mathcal{O}$ as
\begin{align}
    {\cal L}\cdot &= \frac{1}{2}\sigma^2 s^2 \partial^2_{ss}\cdot\\
    {\cal D}\cdot &= \frac{s-v}{t}\partial_v \cdot -  \tilde{a}^*(t,x) \partial_q \cdot
\end{align}
where $x=(s,v,q,c)\in\mathcal{O}$. One can verify that $\mathcal{L}$, $\mathcal{D}$ and $\mathcal{L}+\mathcal{D}$ generate a $C^0$ semi-groups, thus, the solution of \eqref{linearHJBEqua} at $t-\Delta t$ can be represented as
\begin{align}
    P(t-\Delta t, x) = e^{\Delta t (\mathcal{L}+\mathcal{D})} P(t,x)
\end{align}
where $e^{\Delta t (\mathcal{L}+\mathcal{D})}$ denotes the semi-group associated to the parabolic linear PDE \eqref{linearHJBEqua}. A first order approximation of the solution operator is obtained using Baker–Campbell–Hausdorff formula and Lie-Trotter splitting (see \cite{Trotter})
\begin{align} \label{splitting}
    e^{\Delta t (\mathcal{L}+\mathcal{D})} P(t,x) = e^{\Delta t \mathcal{D}}e^{\Delta t \mathcal{L}} P(t,x) + O(\Delta t)
\end{align}
One can also use Strang splitting $e^{\frac{\Delta t}{2} \mathcal{D}}e^{\Delta t \mathcal{L}}e^{\frac{\Delta t}{2} \mathcal{D}}$ to get a second order approximation. The splitting \eqref{splitting} corresponds to solving the parabolic PDE first with generator $\cal L$ and then the first-order transport PDE corresponding to the operator $\cal D$. By using the method of characteristics,  the solution corresponding to $\cal D$ is explicitly given by 
\begin{align}
  e^{\Delta t {\cal D}} Q(t,x) \; = \; Q(t,s,v+{s-v \over t}\Delta t,q + \tilde{a}^*(t,x) \Delta t, c + \tilde{a}^*(t,x) s \Delta t)   
\end{align}   
where $x=(s,v,q,c)\in\mathcal{O}$ and $Q(t,x)=e^{\Delta t \mathcal{L}}P(t,x)$. Finally, we extend $P(t,\cdot)$ to $\mathbb{R}^*_+\times \mathbb{R}^*_+\times \mathbb{R}\times \mathbb{R}$ using boundary conditions.

\bibliographystyle{plain}

\bibliography{biblioRL}

\end{document}